\documentclass[11pt]{article}

\usepackage[letterpaper,margin=1.00in]{geometry}
\usepackage{amsmath, amssymb, amsthm, amsfonts}
\usepackage{thmtools,thm-restate}
\usepackage{bbm}

\usepackage{ifthen}
\usepackage{tikz}
\usetikzlibrary{positioning,decorations.pathreplacing}

\usepackage{cite}
\usepackage{appendix}
\usepackage{graphicx}
\usepackage{xcolor}
\usepackage{algorithm}
\usepackage[noend]{algpseudocode}
\usepackage{epstopdf}
\usepackage{wrapfig}
\usepackage{paralist}
\usepackage{wasysym}
\usepackage{comment}
\usepackage{tcolorbox}

\usepackage{dsfont}

\newcommand{\remove}[1]{}

\usepackage[textsize=tiny]{todonotes}
\newcommand{\michal}[1]{\todo[color=blue!40, author=Michal]{#1}}

\usepackage{framed}
\usepackage[framemethod=tikz]{mdframed}
\usepackage[bottom]{footmisc}
\usepackage{enumitem}
\setitemize{noitemsep,topsep=3pt,parsep=3pt,partopsep=3pt}
\usepackage[font=small]{caption}
\usepackage{xspace}

\newtheorem{theorem}{Theorem}[section]
\newtheorem{lemma}[theorem]{Lemma}
\newtheorem{meta-theorem}[theorem]{Meta-Theorem}
\newtheorem{claim}[theorem]{Claim}
\newtheorem{remark}[theorem]{Remark}

\newtheorem{observation}[theorem]{Observation}

\definecolor{darkgreen}{rgb}{0,0.5,0}
\definecolor{darkblue}{rgb}{0,0,0.6}
\usepackage{hyperref}
\hypersetup{
    unicode=false,          % non-Latin characters in Acrobat’s bookmarks
    colorlinks=true,        % false: boxed links; true: colored links
    linkcolor=darkblue,          % color of internal links (change box color with linkbordercolor)
    citecolor=darkgreen,        % color of links to bibliography
    filecolor=magenta,      % color of file links
    urlcolor=cyan           % color of external links
}
\usepackage[capitalize, nameinlink]{cleveref}
\Crefname{remark}{Remark}{Remarks}
\Crefname{observation}{Observation}{Observations}
\algnewcommand\algorithmicswitch{\textbf{switch}}
\algnewcommand\algorithmiccase{\textbf{case}}

\algdef{SE}[SWITCH]{Switch}{EndSwitch}[1]{\algorithmicswitch\ #1\ \algorithmicdo}{\algorithmicend\ \algorithmicswitch}%
\algdef{SE}[CASE]{Case}{EndCase}[1]{\algorithmiccase\ #1}{\algorithmicend\ \algorithmiccase}%
\algtext*{EndSwitch}%
\algtext*{EndCase}%
\newcommand{\eps}{\varepsilon}

\newcommand{\CONGEST}{$\mathsf{CONGEST}$\xspace}
\newcommand{\LOCAL}{$\mathsf{LOCAL}$\xspace}

\newcommand{\opt}{\mathrm{OPT}}
\newcommand{\mds}{\mathrm{MDS}}
\newcommand{\mvc}{\mathrm{MVC}}
\newcommand{\mfvc}{\mathrm{MFVC}}
\newcommand{\hnode}{\textsc{heavy}}
\newcommand{\lnode}{\textsc{light}}

\DeclareMathOperator{\E}{\mathbb{E}}

\newcommand{\FullOrShort}{full}

\ifthenelse{\equal{\FullOrShort}{full}}{
	
  \newcommand{\fullOnly}[1]{#1}
  \newcommand{\shortOnly}[1]{}
	\renewcommand{\baselinestretch}{1.00}
  
  }{
    \renewcommand{\baselinestretch}{1.00}
    \newcommand{\fullOnly}[1]{}
    \newcommand{\IncludePictures}[1]{}
   
  }

\begin{document}
\date{}
\title{Near-Optimal Distributed Dominating Set \\ in Bounded Arboricity Graphs}
\author{Michal Dory \\ \small{michal.dory@inf.ethz.ch} \\ \small{ETH Zurich} \and Mohsen Ghaffari \\ \small{ghaffari@inf.ethz.ch} \\ \small{ETH Zurich} \and Saeed Ilchi \\ \small{saeed.ilchi@inf.ethz.ch} \\ \small{ETH Zurich}}

\maketitle

\begin{abstract}
We describe a simple deterministic $O( \eps^{-1} \log \Delta)$ round distributed algorithm for $(2\alpha+1)(1 + \eps)$ approximation of minimum weighted dominating set on graphs with arboricity at most $\alpha$. Here $\Delta$ denotes the maximum degree. We also show a lower bound proving that this round complexity is nearly optimal even for the unweighted case, via a reduction from the celebrated KMW lower bound on distributed vertex cover approximation [Kuhn, Moscibroda, and Wattenhofer JACM'16].

Our algorithm improves on all the previous results (that work only for unweighted graphs) including a randomized $O(\alpha^2)$ approximation in $O(\log n)$ rounds [Lenzen and Wattenhofer DISC'10], a deterministic $O(\alpha \log \Delta)$ approximation in $O(\log \Delta)$ rounds [Lenzen and Wattenhofer DISC'10], a deterministic $O(\alpha)$ approximation in $O(\log^2 \Delta)$ rounds [implicit in Bansal and Umboh IPL'17 and Kuhn, Moscibroda, and Wattenhofer SODA'06], and a randomized $O(\alpha)$ approximation in $O(\alpha\log n)$ rounds [Morgan, Solomon and Wein DISC'21].

%We also provide a randomized $O(t \log\Delta)$ round distributed algorithm for any $t \leq \alpha \log^{-1}\alpha$ with expected approximation factor $\alpha + O(\alpha t^{-1})$.
We also provide a randomized $O(\alpha \log\Delta)$ round distributed algorithm that sharpens the approximation factor to $\alpha(1+o(1))$. 
If each node is restricted to do polynomial-time computations, our approximation factor is tight in the first order as it is NP-hard to achieve $\alpha - 1 - \eps$ approximation [Bansal and Umboh IPL'17]. 

\end{abstract}

\bibliographystyle{alpha}

\section{Introduction}

The minimum dominating set (MDS) problem is a classic and central problem in graph algorithms. In this problem, the goal is to construct a minimum weight set of nodes $S$ such that each node is either in $S$ or has a neighbor in $S$. The MDS problem has been widely studied both in the classic centralized setting and in the distributed setting, and it has various applications, for example, clustering and routing in ad-hoc networks. It is well-known that a simple greedy algorithm gives $\ln{(\Delta+1)}$ approximation for the problem for graphs with maximum degree $\Delta$ \cite{johnson1974approximation}, and that it is NP-hard to obtain a $c \ln{\Delta}$-approximation for a suitable constant $c$ \cite{dinur2014analytical}. A similar approximation can be also obtained by efficient distributed algorithms \cite{jia2002efficient,kuhn2005constant,nearsighted,deurer2019deterministic}. In particular, Kuhn, Moscibroda, and Wattenhofer showed a randomized $O(\log{\Delta})$-approximation algorithm that takes $O(\log^2{\Delta})$ rounds in the \CONGEST model, where messages are restricted to $O(\log{n})$ bits, or $O(\log{n})$ rounds in the \LOCAL model, where the message size is unbounded \cite{nearsighted}. There is also a deterministic poly-logarithmic $O(\log{\Delta})$-approximation algorithm for the problem in \CONGEST that is obtained by combining the algorithm of Deurer, Kuhn, and Maus \cite{deurer2019deterministic} with the recent deterministic network decomposition of Rozho\v{n} and Ghaffari \cite{rozhonghaffari20}. If one allows unbounded messages and exponential local computation, one can obtain even $(1+\eps)$-approximation for the problem in poly-logarithmic time in the \LOCAL model using the algorithm of Ghaffari, Kuhn, and Maus \cite{ghaffari2017complexity}.
On the lower bound side, Kuhn, Moscibroda, and Wattenhofer showed that one needs $\Omega(\log{\Delta}/\log{\log{\Delta}})$ rounds or $\Omega(\sqrt{\log{n}/\log{\log{n}}})$ rounds (the minimum of these two lower bounds applies) to get a logarithmic approximation \cite{KMW}.

Since it is NP-hard to obtain better than logarithmic approximation in general graphs, a large body of research focused on finding algorithms with a better approximation in special graph families, such as planar graphs, graphs of bounded expansion, and more (see, e.g., \cite{lenzen2013distributed,akhoondian2018distributed,amiri2019distributed,DBLP:conf/wdag/BonamyCGW21,czygrinow2008fast,lenzen2008can,wawrzyniak2014strengthened}).
One prominent example is the class of \emph{bounded arboricity} graphs, which informally speaking are graphs that are sparse everywhere. The arboricity $\alpha$ of a graph is the minimum number of forests into which its edges can be partitioned. The class of bounded arboricity graphs includes many important graph classes such as planar graphs, graphs of bounded treewidth or genus, and graphs excluding a fixed minor.
Many real-world graphs are sparse and believed to have low arboricity, for example, the World Wide Web graph and graphs representing social networks. This led to extensive study of graph problems in low-arboricity graphs (see, e.g., \cite{goel2006bounded,lenzen2010minimum,bansal2017tight,DBLP:conf/wdag/MorganSW21,eden2020testing,gianinazzi2021parallel,DBLP:conf/wdag/CambusCMU21}).

\subsection{MDS in Bounded Arboricity Graphs}

Lenzen and Wattenhofer showed the first algorithms for MDS in bounded arboricity graphs \cite{lenzen2010minimum}. In particular, they showed a randomized $O(\alpha^2)$-approximation algorithm that takes $O(\log{n})$ time, and a deterministic $O(\alpha \log{\Delta})$-approximation algorithm that takes $O(\log{\Delta})$ time. An $O(\alpha^2)$-approximation can be also obtained deterministically in $O(\log{n})$ time as was shown recently by Amiri \cite{DBLP:journals/corr/abs-2102-08076}. All these algorithms work in the \CONGEST model.
A recent line of work shows $O(\alpha)$-approximation algorithms for the problem. First, a centralized algorithm of Bansal and Umboh gives $(2\alpha+1)$-approximation for the problem \cite{bansal2017tight}.\footnote{The paper claimed a $3\alpha$-approximation, but optimizing the parameters of the algorithm gives a $(2\alpha+1)$-approximation, as was observed by Dvo{\v{r}}{\'a}k \cite{dvovrak2019distance}.} They also show that it is NP-hard to obtain an $(\alpha-1-\epsilon)$-approximation for the problem. The algorithm of Bansal and Umboh is based on LP-rounding, and can be implemented efficiently in the \CONGEST model using an algorithm for approximating the LP. This leads to a deterministic $O(\log^2{\Delta}/\eps^4)$-round $(2\alpha+1)(1+\epsilon)$-approximation algorithm using the $(1+\epsilon)$-approximation algorithm of Kuhn, Moscibroda, and Wattenhofer for approximating the LP \cite{nearsighted}.
A recent combinatorial algorithm for MDS in bounded arboricity graphs was shown by Morgan, Solomon and Wein \cite{DBLP:conf/wdag/MorganSW21}; they obtain a randomized $O(\alpha \log{n})$-round $O(\alpha)$-approximation algorithm for the problem in \CONGEST. All the above algorithms solve the unweighted version of the problem. Lastly, a very recent centralized algorithm of Sun \cite{sun2021improved} gives an $(\alpha+1)$-approximation for the weighted version of the problem. This algorithm however seems inherently sequential and does not seem to translate to an efficient distributed algorithm (see \cref{sec:our_techniques} for a more detailed discussion).

\subsection{Our Contribution}

Our first contribution is showing that $O(\alpha)$-approximation for weighted MDS can be obtained in just $O(\log{\Delta})$ rounds in the \CONGEST model. In particular, we show a simple deterministic algorithm that gives the following. For this and all the other algorithms in this paper, we assume that both $\Delta$ and $\alpha$ are known to all nodes. Please see~\cref{remark:unknown-delta} and~\cref{remark:unknown-alpha} for a discussion on the setting where $\Delta$ and $\alpha$ are unknown.

\begin{restatable}{theorem}{ThmDetMDS}\label{thm:det-ds}
For any $0 < \eps < 1$, there is a deterministic $(2\alpha + 1)(1 + \eps)$-approximation algorithm for the minimum weighted dominating set problem in graphs with arboricity at most $\alpha$. The algorithm runs in $O\left(\frac{\log (\Delta/\alpha)}{\eps}\right)$ rounds in the \CONGEST model.
\end{restatable}

Our algorithm is faster compared to the two previous $O(\alpha)$-approximation algorithms that take $O(\log^2{\Delta}/\eps^4)$ rounds \cite{bansal2017tight,nearsighted} and $O(\alpha \log{n})$ rounds \cite{DBLP:conf/wdag/MorganSW21}, and its $O(\log \Delta)$ complexity is nearly optimal, as we will discuss later when describing lower bounds. Its approximation ratio of $(2\alpha+1)(1+\eps)$ matches the best approximation that was previously obtained by a distributed algorithm. It is a deterministic algorithm, where the previous $O(\alpha \log{n})$ round algorithm was randomized. Moreover, to the best of our knowledge, our algorithm is the first distributed algorithm that solves the \emph{weighted} version of the problem.

\paragraph{Improved Approximation}

We also show a randomized algorithm with an improved approximation of $\alpha(1+o(1))$, giving the following.

\begin{restatable}{theorem}{ThmRandMDS}
\label{thm:rand-mds}
For any $1 \leq t \leq \frac{\alpha}{\log \alpha}$, there is a randomized algorithm with expected $(\alpha + O(\frac{\alpha}{t}))$-approximation factor for the minimum weighted dominating set problem in graphs with arboricity at most $\alpha$. The algorithm runs in $O(t \log \Delta)$ rounds in the \CONGEST model.
\end{restatable}

In particular, by setting $t = \frac{\alpha}{\log \alpha}$, we can get $(\alpha+O(\log{\alpha}))$-approximation in $O(\alpha \log{\Delta})$ time. For algorithms that are only using polynomial time computations, this approximation factor is tight in the first order as it is NP-hard to achieve $\alpha - 1 - \eps$ approximation \cite{bansal2017tight}.

%\saeed{I thought that maybe it's good to mention this improvement in the introduction. please check it.}
As a byproduct of our randomized algorithm, we improve on the approximation factor for the minimum dominating set problem on general graphs with maximum degree $\Delta$. Previously, the best approximation factor for the problem in $O(k^2)$ (for a parameter $k$) rounds was due to the work of Kuhn, Moscribroda, and Wattenhofer~\cite{nearsighted} where they provided a randomized algorithm with expected approximation factor $O(k\Delta^{\frac{2}{k}} \log \Delta)$. We drop the $O(\log \Delta)$ term in their result.

\begin{restatable}{theorem}{ThmRandGeneralMDS}
\label{thm:dsgeneral}
For any $k$, there is randomized algorithm that computes a weighted dominating set with expected approximation factor at most $\Delta^{\frac{1}{k}}(\Delta^{\frac{1}{k}} + 1)(k+1) \cdot = O(k\Delta^{\frac{2}{k}})$ in $O(k^2)$ rounds in the \CONGEST model.
\end{restatable}

\paragraph{Lower Bound} Our algorithms provide $O(\alpha)$-approximation in $O(\log{\Delta})$ rounds. A natural question is whether a logarithmic dependence on $\Delta$ is needed in the time complexity. While in general graphs it is known that $\Omega(\log{\Delta}/\log{\log{\Delta}})$ rounds are required for obtaining a constant or logarithmic approximation \cite{KMW}, in special graph families faster algorithms are known. For example, in planar graphs one can obtain $O(1)$-approximation in $O(1)$ rounds in the \LOCAL model \cite{lenzen2008can,wawrzyniak2014strengthened}. Also, in the special case of trees that have arboricity 1, a trivial algorithm that takes all non-leaf nodes gives a 3-approximation for unweighted MDS (see \cref{app_miss_proof} for the proof).
Interestingly, we show that as soon as the arboricity is increased from $1$ to $2$, the locality of the minimum dominating set approximation problem changes radically, and any constant or poly-logarithmic approximation needs $\Omega\left(\frac{\log \Delta}{\log\log \Delta}\right)$ time, even in the \LOCAL model where the size of messages is unbounded. 

\begin{restatable}{theorem}{ThmLB}\label{thm:lower-bound}
Any distributed algorithm that computes any constant or poly-logarithmic approximation of the minimum dominating set on graphs of arboricity 2 requires $\Omega\left(\frac{\log \Delta}{\log\log \Delta}\right)$ rounds in the \LOCAL model.
\end{restatable}

Hence, the $O(\log \Delta)$ round complexity of our algorithms is nearly optimal.

\subsection{Our Techniques} \label{sec:our_techniques}
%\michal{feel free to edit this part, in particular feel free to add more details on any part.}
At a high-level, our algorithms construct a dominating set in two steps. In the first step, we construct a partial dominating set $S$ that has the following nice properties. First, the weight of the set $S$ is a good approximation for the optimal dominating set. Second, the nodes $T$ that are undominated by $S$ have a nice structure that allows us to find efficiently a dominating set for $T$ with a good approximation guarantee. In the second step, we construct a dominating set for $T$. In fact, in our first algorithm, we show that if for each node in $T$ we add one node to the dominating set, we already obtain a $(2\alpha+1)(1+\eps)$-approximation. In our second algorithm, we show how to exploit the structure of $T$ to obtain a better approximation.
Our algorithm for constructing the partial dominating set is inspired by the primal-dual method. To explain the idea, we first describe in \cref{sec:unweighted_MDS} a simpler variant of our algorithm that works for unweighted graphs. Next, in \cref{sec:weightedMDS}, we generalize the algorithm for weighted graphs, and also show our randomized algorithm with improved approximation. Our lower bound appears in \cref{sec:LB}. 

\paragraph{Comparison to \cite{sun2021improved}} We remark that in a very recent independent work \cite{sun2021improved}, the author uses the primal-dual method to obtain approximation algorithm for MDS in bounded arboricity graphs in the centralized setting. This algorithm however does not seem to translate to an efficient distributed algorithm, as it has a reverse-delete step that makes it inherently sequential. In this step, the algorithm goes over all the nodes that were added to the dominating set $S$ in reverse order and removes them from $S$ if it is still a valid dominating set. This part is crucial for obtaining a good approximation ratio, and the analysis crucially relies on this part.
While in our algorithm we also use the primal-dual method, we use it in a different way, and only to construct a partial solution for the problem.

%----

\section{Preliminaries}

The input graph is $G = (V, E)$ with $n$ nodes, $m$ edges, maximum degree $\Delta$, and arboricity $\alpha$. For each node $v \in V$, let $N_v$ be the set of neighbors of $v$ and $N^{+}_v = \{v\} \cup N_v$. For a set of nodes $S\subseteq V$, let $N^{+}_S = \bigcup_{v \in S} N^{+}_v$ be set of nodes that are dominated by $S$. Let $w_v$ be the weight of $v$ and for set $S \subseteq V$, let $w_S = \sum_{v \in S} w_v$ be the total weight of the set $S$. %Moreover, for each node $v \in V$, let $\tau_v = \min_{u \in N^{+}_v} w_u$ be the minimum weight of a node that can dominate $v$. 
We assume all the weights are positive integers and are bounded by $n^c$ for some constant $c$.

Our algorithms are inspired by the primal-dual method.
We associate a packing value $0 \leq x_v$ to each node such that for any node $u$, the value $X_u = \sum_{v \in N^{+}_u} x_v \leq w_u$. From weak duality, we have the following.

\begin{lemma}\label{claim_lb_opt}
For any feasible packing, $\sum_{v \in V} x_v \leq \opt$ where $\opt$ is the weight of the minimum dominating set of $G$.
\end{lemma}

\begin{proof}
Let $S^*$ be a dominating set of $G$ with weight $\opt$. We have:
\begin{equation*}
    \opt = \sum_{v \in S^*} w_v
    \geq \sum_{v\in S^*} X_v
    \geq \sum_{u\in V} x_u 
\end{equation*}
The last inequality comes from the fact that $S^*$ is a dominating set. Hence, each node $u$ contributes to at least one of $X_v$ for $v \in S^*$.
\end{proof}

\paragraph{Model} Our algorithms work in the standard \CONGEST model. We have a communication network with $n$ nodes that is identical to the input graph. Nodes communicate with each other by sending $O(\log{n})$ bit messages in synchronous rounds. In the beginning, each node only knows its own weight and set of neighbors. At the end of the algorithm, each node should know if it is part of the constructed dominating set.
Our lower bound works even for the more powerful \LOCAL model where the size of messages is unbounded. 

\section{Algorithm for Unweighted MDS} \label{sec:unweighted_MDS}

As a warm-up, we start by describing a simpler variant of our algorithm that works for unweighted graphs, showing the following.

\begin{theorem}
\label{thm:connected-subgraph}
For any $0 < \eps < 1$, there is a deterministic $(2\alpha+1)(1 + \eps)$-approximation algorithm for the minimum dominating set problem in unweighted graphs with arboricity at most $\alpha$ that runs in $O\left(\frac{\log (\Delta/\alpha)}{\eps}\right)$ rounds in the \CONGEST model.
\end{theorem}

Our approach is inspired by the primal-dual method. Each node $v$ has a packing value $x_v$, such that during the algorithm for each node, the value $X_v = \sum_{u \in N^{+}_v} x_u \leq 1$. From \cref{claim_lb_opt} we have that $\sum_{v \in V} x_v \leq \opt$, where $\opt$ is an optimal solution.
Our algorithm first builds a partial dominating set $S$ with the following properties. 

\begin{lemma}\label{lem:part_dom}
For any $0 < \eps < 1$, there is a deterministic algorithm that takes $O\left(\frac{\log (\Delta/\alpha)}{\eps}\right)$ rounds and outputs a partial dominating set $S \subseteq V$ along with packing values $\{x_v\}_{v \in V}$ such that
\begin{enumerate}
    \item $|S| \leq (2\alpha + 1)(1+\eps) \sum_{v \in N^{+}_S} x_v$.
    \item For each node $v \not \in N^{+}_S$, we have $x_v \geq \frac{1}{(2\alpha+1)(1+\eps)}$.
\end{enumerate}
\end{lemma}

Let $T= V \setminus N^{+}_S$ be the set of undominated nodes. We next show that adding the nodes of $T$ to the dominating set $S$ results in the desired approximation. 

\begin{claim}
The set $S \cup T$ is a dominating set of size at most $(2\alpha+1)(1+\eps)\opt$. 
\end{claim}

\begin{proof}
The set $S \cup T$ is clearly a dominating set, as we added to $S$ all the undominated nodes. From \cref{lem:part_dom}, we have that  $$|S| \leq (2\alpha + 1)(1+\eps) \sum_{v \in N^{+}_S} x_v,$$ $$|T| = \sum_{v \in T} 1 \leq (2 \alpha +1 )(1+\eps) \sum_{v \in T} x_v.$$
Hence, we get that $|S \cup T| \leq (2\alpha+1)(1+\eps)\sum_{v \in V} x_v \leq (2\alpha+1)(1+\eps) \opt$, where the last inequality follows from \cref{claim_lb_opt}.
\end{proof}

To complete the proof, our goal is to prove \cref{lem:part_dom}.

\begin{proof}[Proof of \cref{lem:part_dom}]
For all nodes, we initialize $x_v = \frac{1}{\Delta + 1}$. Let $r$ be the integer such that
\begin{equation*}
    (1 + \eps)^r \frac{1}{\Delta + 1} \leq \frac{1}{(2\alpha + 1)(1+\eps)} < (1 + \eps)^{r+1} \frac{1}{\Delta + 1}
\end{equation*}

\paragraph{Algorithm for Partial Dominating Set} At the beginning, all the nodes are unmarked. Then, we run the following procedure for $r+1$ iterations:  Per iteration, we run the following on each node $v$. All nodes run each line simultaneously.
\begin{enumerate}
    \item Compute $X_v = \sum_{u \in N^{+}_v} x_u$.
    \item If $X_v \geq \frac{1}{1 + \eps}$, add $v$ to $S$ and mark all the nodes in $N^{+}_v$.
    \item If $v$ is not marked, set $x_v \gets x_v (1 + \eps)$.
\end{enumerate}

\begin{lemma}
For all nodes, $X_v$ is always at most 1. 
\end{lemma}
\begin{proof}
At the start of the first iteration, $X_v \leq \frac{|N^{+}_v| }{\Delta + 1} \leq 1$. At the beginning of iteration $i \geq 2$, if $v$ is not in $S$, it means $X_v < \frac{1}{1 + \eps}$ in iteration $i-1$ and so $X_v <1$ in this iteration. If $v$ is in $S$, then $X_v$ is not changed after that and so $X_v \leq 1$.
%For the second claim, let $S^*$ be a dominating set of $G$ with size $\opt$. We have:
%\begin{equation*}
%    \opt = \sum_{v \in S^*} 1
%    \geq \sum_{v\in S^*} X_v
%    \geq \sum_{u\in V} x_u 
%\end{equation*}
%The last line comes from the fact that $S^*$ is a dominating set. Hence, each node $u$ contributes to at least one of $X_v$ for $v \in S^*$.
\end{proof}

To bound the size of $S$, we use the following property of bounded arboricity graphs. 

\begin{observation}\label{obs_bounded_arb}
Let $G$ be a graph with arboricity at most $\alpha$, then the edges of $G$ can be oriented such that the out-degree of each node is at most $\alpha$.\footnote{All of our algorithms work as long the input graph is orientable with the maximum out-degree of at most $\alpha$. Hence, our results can be applied to the slightly larger class of graphs that are decomposable to at most $\alpha$-pseudoforests. A psuedoforest is a graph in which every connected component has at most one cycle.}
\end{observation}

Observation \ref{obs_bounded_arb} follows from the fact that the edges of the graph can be partitioned into $\alpha$ forests, and in each one of them we can orient the edges with out-degree one, by fixing a root in each tree and orienting the edges towards the root.
We fix one of those orientations. We emphasize that this orientation is used only in the analysis, and we do not construct it in the algorithm. For each node $v$, let $N^{\mathrm{in}}_v$ be the set of incoming neighbors and $N^{\mathrm{out}}_v$ be the set of outgoing neighbors of $v$ with respect to this fixed orientation. Note that for all nodes in $N^{+}_S$, the packing value is increased at most $r$
times. Hence, from the choice of $r$, we have that $x_v \leq \frac{1}{(2\alpha + 1)(1+\eps)}$ for $v \in N^{+}_S$. For each $v \in S$, we have (all the values $x_v$ are considered at the end of the algorithm):
\begin{align*}
    X_v = \sum_{u \in N^{\mathrm{in}}_v} x_u + x_v + \sum_{u' \in N^{\mathrm{out}}_v} x_{u'}
    &\geq \frac{1}{1 + \eps}\\
    \Rightarrow \sum_{u \in N^{\mathrm{in}}_v} x_u
    &\geq \frac{1}{1 + \eps} - \frac{\alpha + 1}{(2\alpha + 1)(1+\eps)}\\
    &\geq \frac{\alpha}{(2\alpha+1)(1+\eps)}
\end{align*}

Let $\lambda=\frac{\alpha}{(2\alpha+1)(1+\eps)}$. From the above, we have that $\frac{1}{\lambda} \sum_{u \in N^{\mathrm{in}}_v} x_u \geq 1$.
We can bound the size of $S$ as follows:
\begin{align*}
    |S| = \sum_{v \in S} 1
    \leq \frac{1}{\lambda} \sum_{v \in S} \sum_{u \in N^{\mathrm{in}}_v} x_u
    &\leq \frac{1}{\lambda} \sum_{u \in N^{+}_S} x_u \sum_{v \in S} \mathbb{I}[u \in N^{\mathrm{in}}_v]\\
    &\leq \frac{\alpha}{\lambda} \sum_{u \in N^{+}_S} x_u\\
    &\leq (2\alpha+1)(1+\eps) \opt
\end{align*}

Note that we bound $\sum_{v \in S}\mathbb{I}[u \in N_{v}^{in}]$ with $\alpha$. This is because the out-degree of $u$ is at most $\alpha$, so $u$ can be an incoming neighbor of at most $\alpha$ nodes.

To conclude the proof, we should show that for each node $v \not \in N^{+}_S$, we have $x_v \geq \frac{1}{(2\alpha+1)(1+\eps)}$. This follows from the choice of $r$, as for nodes $v \not \in N^{+}_S$ we increase the packing value $r+1$ times.

%The only remaining part is that each node $v$ has at most $2\alpha$ neighbors in $T = V \setminus (S \cup N(s))$. After the final iteration, all the nodes in $T$ have weight more than $\frac{1}{2\alpha}$ \michal{also, in these lines, check the number $2\alpha$ and explain where it comes from: what is the invariant when the algorithm stops.}. On the other hand, we know $W_v$ is at most $1$ for each node $v$. Hence, $|N(v) \cap T| < 2\alpha$.

Each iteration can be implemented in $O(1)$ rounds in the \CONGEST model. In total, $O(r) = O\left(\frac{\log(\Delta/\alpha)}{\eps}\right)$ rounds.
\end{proof}

%\begin{proof}
%Our algorithm computes a partial dominating set $S$ of size at most $(2 + \eps) \alpha \cdot \opt$ such that any node has at most $2\alpha$ \michal{check numbers, I think currently maybe this should be $2\alpha +2$?} neighbors in the set $T=V \setminus (S \cup N(S))$ of nodes that are not dominated by $S$. This implies $\left | T \right | \leq 2\alpha \cdot \opt$ \michal{here I think there should be +1 to the number that appears in the previous line.}. Hence, $S \cup T$ is a dominating set of size $(4 + \eps)\alpha \cdot \opt$.
%\end{proof}

\section{Algorithm for Weighted MDS}\label{sec:weightedMDS}

We next show a generalized version of our algorithm from \cref{sec:unweighted_MDS}. First, the algorithm works for weighted graphs. Second, we allow for different trade-offs between the two dominating sets we compute in the algorithm. Recall that in \cref{sec:unweighted_MDS} we start by computing a partial dominating set $S$, and then we add to it additional set $S'$ to dominate the rest of the nodes. Previously we just constructed $S'$ as the set of all undominated nodes, but we will see in \cref{sec:randomized} an algorithm that exploits the structure of undominated nodes to get an improved approximation for this part. To get a better approximation for the whole algorithm, we can stop the algorithm for computing a partial dominating set earlier and get an improved approximation for the first part as well. We start by presenting our general scheme, which allows us to get a deterministic algorithm for weighted graphs with the same guarantees obtained in \cref{sec:unweighted_MDS}, and then we show a randomized algorithm that can obtain an improved approximation.

\subsection{Deterministic Algorithm}

We again follow the primal-dual method. We associate a packing value $0 \leq x_v$ to each node such that for any node $u$, the value $X_u = \sum_{v \in N^{+}_u} x_v \leq w_u$. %We say $u$ is $\eps$-tight if $(1-\eps) w_u \leq X_u$. 
For each node $v \in V$, let $\tau_v = \min_{u \in N^{+}_v} w_u$ be the minimum weight of a node that can dominate $v$.

\begin{lemma}
\label{lem:partial-ds}
For any $0 < \eps < 1$ and $0 < \lambda < \frac{1}{(\alpha + 1)(1 + \eps)}$, there is a deterministic algorithm that outputs a partial dominating set $S \subseteq V$ along with packing values $\{x_v\}_{v \in V}$ with the following properties:
\begin{enumerate}[label=(\alph*)]
    \item \label{item:partial-ds-prop-a} $w_S \leq \alpha (\frac{1}{1 + \eps} - \lambda(\alpha + 1))^{-1} \sum_{v \in N^{+}_S} x_v$.
    \item \label{item:partial-ds-prop-b} For each undominated node $v \not \in N^{+}_S$, its packing value $x_v$ is at least $\lambda \tau_v$.
\end{enumerate}
The algorithm runs in $O\left(\frac{\log(\Delta \lambda)}{\eps}\right)$ rounds in the \CONGEST model.
\end{lemma}

\ThmDetMDS*

%\begin{theorem}
%\label{thm:det-ds}
%For any $0 < \eps$, there is a deterministic $(2\alpha + 1)(1 + \eps)$-approximation algorithm for the minimum weighted dominating set problem in graphs with arboricity at most $\alpha$. The algorithm runs in $O\left(\frac{\log (\Delta/\alpha)}{\eps}\right)$ rounds in the \CONGEST model.
%\end{theorem}

\begin{proof}
We run the algorithm of \cref{lem:partial-ds} with the same $\eps$ and with $\lambda$ equals to $\frac{1}{(2\alpha + 1)(1 + \eps)}$. We want to find a set $S'$ such that $S \cup S'$ is a dominating set. For this, we go over all the undominated nodes $v \not \in N^{+}_S$, and add a node from $N^{+}_v$ with weight $\tau_v$ to $S'$. Clearly, $S \cup S'$ is a dominating set and its weight can be bounded as follows: 

\begin{align*}
    w_{S \cup S'}= w_S + \sum_{v \in V \setminus N^{+}_S} \tau_v
    &\leq \alpha (\frac{1}{1 + \eps} - \lambda(\alpha + 1))^{-1} \sum_{v' \in N^{+}_S} x_{v'} + \sum_{v \in V \setminus N^{+}_S} \frac{x_v}{\lambda}\\
    &\leq (2\alpha + 1)(1 + \eps) \sum_{v \in V} x_v\\
    &\leq (2\alpha + 1)(1 + \eps) \cdot \opt
\end{align*}

In the first inequality, we use property \ref{item:partial-ds-prop-a} of \cref{lem:partial-ds} to bound the first term, and we invoked property \ref{item:partial-ds-prop-b} of this lemma, i.e. $\lambda \tau_v \leq x_v$, to bound the second term.
\end{proof}

\begin{proof}[Proof of \cref{lem:partial-ds}]
For each node $v$, initialize $x_v$ to $\frac{\tau_v}{\Delta + 1}$. This gives us a feasible packing since for each node $u$:
\begin{equation*}
X_u = \sum_{v \in N^{+}_u} x_v = \sum_{v \in N^{+}_u} \frac{\tau_v}{\Delta + 1} \leq \sum_{v \in N^{+}_u} \frac{w_u}{\Delta + 1} \leq w_u
\end{equation*}
If $\lambda < \frac{1}{\Delta + 1}$, we can satisfy the two required properties by setting $S$ to the empty set. Assume $\lambda \geq \frac{1}{\Delta + 1}$ and let $r \geq 1$ be the integer that
\begin{equation*}
    (1 + \eps)^{r-1} \frac{1}{\Delta + 1} \leq \lambda < (1 + \eps)^r \frac{1}{\Delta + 1}
\end{equation*}
We run the following procedure for $r$ iterations.
\begin{enumerate}
    \item For each node $u$, compute $X_u = \sum_{v \in N^{+}_u} x_v$.
    \item If $X_u \geq \frac{w_u}{1 + \eps}$, add $u$ to $S$.%Add all the $(1 - \frac{1}{1 + \eps})$-tight nodes to $S$. 
    \item For each undominated node $v \not \in N^{+}_S$, set $x_v \gets x_v (1 + \eps)$.
\end{enumerate}
\begin{observation}
Through the algorithm, $\{x_v\}_{v\in V}$ is always a feasible packing.
\end{observation}
\begin{observation}
\label{obs:partial-ds-packing-bounds}
In the end, the packing value of each undominated node, i.e. nodes in $V \setminus N^{+}_S$, is strictly greater than $\lambda \tau_v$ and the packing value of each dominated node is at most $\lambda \tau_v$.
\end{observation}
\begin{proof}
Let $v \in V \setminus N^{+}_S$. Its packing value is multiplied by $1 + \eps$ in all the $r$ iterations. So its final value is $(1 + \eps)^r \frac{\tau_v}{\Delta + 1} > \lambda \tau_v$. If $v \in N^{+}_S$, it is multiplied by $1 + \eps$ at most $r-1$ times. So its final value is at most $(1 + \eps)^{r-1} \frac{\tau_v}{\Delta + 1} \leq \lambda \tau_v$.
\end{proof}

To bound $w_S$, we use \cref{obs_bounded_arb} and orient the edges of $G$ such that the out-degree of each node is at most $\alpha$. The orientation is used only for the analysis.
%To bound $w_S$, we use the fact that $G$ has bounded arboricity. Edges of $G$ can be oriented such that the out-degree of each node is at most $\alpha$. Fix one of those orientations. We emphasize that this orientation is used only in the analysis. 
For each node $v$, let $N^{\mathrm{in}}_v$ be the set of incoming neighbors and $N^{\mathrm{out}}_v$ be the set of outgoing neighbors of $v$ in this fixed orientation.

Consider the packing values at the end of the algorithm. Note that through the algorithm, we freeze the packing value of a node as soon as it gets dominated. With this, we can write the following for each node $u \in S$:

\begin{gather*}
    X_u = \sum_{v \in N^{\mathrm{in}}_u} x_v + x_u + \sum_{v' \in N^{\mathrm{out}}_u} x_{v'}
    \geq \frac{w_u}{1 + \eps}\\
    \Rightarrow \sum_{v \in N^{\mathrm{in}}_u} x_v
    \geq \frac{w_u}{1 + \eps} - \lambda \tau_u - \sum_{v' \in N^{\mathrm{out}}_u} \lambda \tau_{v'}
    \geq w_u (\frac{1}{1 + \eps} - \lambda(\alpha + 1))
\end{gather*}

This implies:
\begin{align}
\label{eq:weight-of-s}
\begin{split}
    w_S = \sum_{u \in S} w_u
    &\leq \sum_{u \in S} (\frac{1}{1 + \eps} - \lambda(\alpha + 1))^{-1} \sum_{v \in N^{\mathrm{in}}_u} x_v\\
    &\leq (\frac{1}{1 + \eps} - \lambda(\alpha + 1))^{-1} \sum_{v \in N^{+}_S} x_v \sum_{u\in S} \mathbb{I}[v \in N^{\mathrm{in}}_u]\\
    &\leq \alpha (\frac{1}{1 + \eps} - \lambda(\alpha + 1))^{-1} \sum_{v \in N^{+}_S} x_v
\end{split}
\end{align}
In the last inequality, we upper bound $\sum_{u \in S}\mathbb{I}[v \in N_{u}^{in}]$ with $\alpha$. This is because out-degree of $v$ is at most $\alpha$, so $v$ can be an incoming neighbor of at most $\alpha$ nodes.

The bound of \cref{eq:weight-of-s} along with \cref{obs:partial-ds-packing-bounds} guarantee property \ref{item:partial-ds-prop-a} and property \ref{item:partial-ds-prop-b}. The only remaining component is the round complexity. Note that each iteration of the procedure runs in $O(1)$ rounds in the \CONGEST model. So in total, there are $O(r) = O(\log_{1 + \eps} \Delta \lambda) = O\left(\frac{\log (\Delta \lambda)}{\eps}\right)$ many rounds.
\end{proof}

\begin{remark}[Unknown $\Delta$]
\label{remark:unknown-delta}
We can transform the algorithm of~\cref{thm:det-ds} into one that works in the setting where $\Delta$ is unknown. Recall that the algorithm has two phases. In the first phase, it computes a partial dominating set $S$ by applying the algorithm of~\cref{lem:partial-ds} with $\lambda$ being $\frac{1}{(2\alpha + 1)(1 + \eps)}$. Then, in the second phase, for each node $v$ that is not dominated by $S$, we add a node with weight $\tau_v$ from $N^+_v$ to the final dominating set. To convert our algorithm, first we initialize the packing value $x_v$ of each node $v$ with $$\frac{\tau_v}{\max_{u \in N^{+}_v} |N^{+}_u|}$$ rather than with $\frac{\tau_v}{\Delta + 1}$. Next, we run the iterations of the algorithm of~\cref{lem:partial-ds} similarly with one extra step at the beginning of each iteration. In this extra step, each undominated node $v$ with a packing value strictly larger than $\lambda \tau_v$ adds a node with weight $\tau_v$ in its neighborhood to the final dominating set. Observe that after $O(\frac{\log \Delta}{\eps})$ rounds, all nodes are dominated and the approximation analysis goes through similarly. Intuitively, since a node cannot decide locally when the first phase terminates, we add this extra step for each iteration to simulate its effect.
\end{remark}

\begin{remark}[Unknown $\alpha$]
\label{remark:unknown-alpha}
When $\alpha$ is unknown, we are not aware of a way to keep the algorithm's round complexity independent of $n$, while preserving the approximation factor. However, we can find a dominating set with approximation factor $(2\alpha + 1)(2 + \eps)$ in $O(\frac{\log n}{\eps})$ rounds (note that here we assume that all nodes know $n$). For that, first we apply the orientation algorithm of Barenboim and Elkin~\cite{barenboim2010sublogarithmic} to find an orientation of edges where the out-degree of each node is at most $(2 + \eps)\alpha$. This algorithm runs in $O(\frac{\log n}{\eps})$ rounds. Next, each node $v$ computes a local approximation of arboricity, denoted by $\hat{\alpha}_v$, for itself which is the maximum out-degree of the nodes in $N^+_v$. To find the partial dominating set $S$, each node initializes its packing value with $\frac{1}{n+1}$. We run the algorithm of~\cref{lem:partial-ds} where each node has its own $\lambda$, denoted by $\lambda_v$, and which equals $\lambda_v = \frac{1}{(2\hat{\alpha} + 1)(1 + \eps)}$. Similar to~\cref{remark:unknown-delta}, we add an extra step in the beginning of each iteration on which any undominated node $v$ with a packing value more than threshold $\lambda_v\tau_v$, adds a node with weight $\tau_v$ in $N^+_v$ to the final dominating set. After $O(\frac{\log n}{\eps})$ iterations, all nodes are dominated and the approximation factor is $(2\alpha + 1)(2 + O(\eps))$.
\end{remark}

\subsection{Randomized Algorithm}\label{sec:randomized}
In the previous section, to extend our partial dominating set $S$ to a dominating set, we simply added one node to $S$ for each undominated node and this introduced a factor $2$ in the approximation factor. Here, we show how we can get $\alpha + O(\log \alpha)$ approximation, but we need $O(\frac{\alpha}{\log \alpha} \log \Delta)$ rounds rather than $O(\log \Delta)$ rounds. The algorithm also becomes randomized.

To reduce the approximation factor, we can leverage property \ref{item:partial-ds-prop-b} of \cref{lem:partial-ds}. To explain the intuition, we focus first on the unweighted case.
If the problem is unweighted (so $\tau_v$ is 1 for all the nodes), then property  \ref{item:partial-ds-prop-b} implies that each node has at most $\lambda^{-1}$ undominated neighbors. The reason is that $X_u = \sum_{v \in N^{+}_u} x_v \leq 1$ for all $u$. Now since any node $v \not \in N^{+}_S$ has $x_v \geq \lambda$, there can only be at most $\lambda^{-1}$ undominated nodes in $N^{+}_u$.
So, dominating the set of nodes that are not in $N^+_{S}$ is a set cover problem with maximum set size $\lambda^{-1}$ which can be approximated with a factor of $O(\log \lambda^{-1})$ in $O(\log \lambda^{-1} \log \Delta)$ rounds according to~\cite{nearsighted} (Note that an undominated node can be a neighbor of many dominated nodes, so the frequency in the set cover instance can be as large as $\Delta + 1$). Recall that the size of $S$ is bounded by $\alpha(\frac{1}{1 + \eps} - \lambda(\alpha + 1))^{-1} \cdot \opt$. On the other hand, for extending $S$, we add $O(\log \lambda^{-1} \cdot \opt)$ many nodes. If we set $\lambda = \Theta(\frac{\log \alpha}{\alpha^2})$ and $\eps = \Theta(\frac{\log \alpha}{\alpha})$, with straightforward calculations we can show that the expected size of the final dominating set is $(\alpha + O(\log \alpha)) \cdot \opt$ and the algorithm takes $O(\frac{\alpha}{\log \alpha} \log \Delta)$ many rounds. The bottleneck for the round complexity is the first phase when we construct $S$ as there the number of rounds depends linearly on $\frac{1}{\eps}$.

The weighted case is more subtle as there is no bound on the set size. Now \cref{lem:partial-ds} implies that $ x_v \geq \lambda \tau_v$, which is a different value for every $v$, hence we cannot argue anymore that each node has at most $\lambda^{-1}$ undominated neighbors.  We are not aware of a result in the literature that we can use as a black box or with a clean reduction for this case, but we still want to exploit property \ref{item:partial-ds-prop-b} of \cref{lem:partial-ds} to get an improved approximation. To do so, we devise a simple iterative randomized algorithm for this case. Towards resolving this, we also improve on the results of~\cite{kuhn2005constant, nearsighted} for solving the dominating set problem on general graphs. There, they presented a $O(k^2)$ rounds randomized algorithm with expected $O(k\Delta^{\frac{2}{k}} \log \Delta)$ approximation factor for the dominating set problem. We shave off the factor $\log \Delta$ in their bound as it is stated in \cref{thm:dsgeneral}.

Our algorithm of extending $S$ in full generality is stated in the following lemma.

\begin{lemma}
\label{lem:small-sets}
For $0 < \lambda$ and $1 < \gamma$, there is a randomized algorithm that given the output of \cref{lem:partial-ds}, it finds $S' \subseteq V$ such that $S \cup S'$ is a dominating set and $\E[w_{S'}] = \gamma(\gamma + 1) \lceil\log_{\gamma}\lambda^{-1} \rceil \cdot \opt$. The algorithm runs in $O(\log_{\gamma}\lambda^{-1}\log_\gamma \Delta)$ rounds in the \CONGEST model. 
\end{lemma}

By optimizing the parameters in \cref{lem:partial-ds} and \cref{lem:small-sets}, we can prove the claim of \cref{thm:rand-mds}. The theorem is restated below.

\ThmRandMDS*

%\begin{theorem}
%\label{thm:rand-mds}
%For any $1 \leq t \leq \frac{\alpha}{\log \alpha}$, there is a randomized algorithm with expected $(\alpha + O(\frac{\alpha}{t}))$-approximation factor for the minimum weighted dominating set problem in graphs with arboricity at most $\alpha$. The algorithm runs in $O(t \log \Delta)$ rounds in the \CONGEST model.
%\end{theorem}

\begin{proof}
We first execute the algorithm of \cref{lem:partial-ds} and then \cref{lem:small-sets} with a suitable set of parameters such that $w_S \leq \alpha(1 + \frac{1}{t})$ and $w_{S'} = O(\frac{\alpha}{t})$. Since $S \cup S'$ is a dominating set, this gives us the desired approximation factor. For the parameters, we set $\eps = \frac{1}{4t}$, $\lambda = \frac{\eps}{\alpha + 1}$, and $\gamma = \max(2, \alpha^{\frac{1}{2t}})$. To bound $w_S$, note that:%\michal{is there some intuitive description for the choice of the parameters?}
\begin{align*}
    w_S
    &\leq \alpha(\frac{1}{1 + \eps} - \lambda(\alpha + 1))^{-1} \cdot \opt\\
    &\leq \alpha(1 - 2\eps)^{-1} \cdot \opt\\
    &\leq \alpha(1 + 4\eps) \cdot \opt\\
    &\leq (\alpha + \frac{\alpha}{t}) \cdot \opt\\
\end{align*}
In the second last inequality, we use the fact that $\eps = \frac{1}{4t} \leq \frac{1}{4}$. To bound $w_{S'}$, first note that $\lambda^{-1} = 4t (\alpha + 1) = O(\alpha^2)$ because $t \leq \frac{\alpha}{\log \alpha}$. If $\gamma$ is $2$, we have:
\begin{equation*}
    \E[w_{S'}]
    = O(\gamma^2 \log_{\gamma}\lambda^{-1} \cdot \opt)
    = O(\log \alpha) \cdot \opt
    = O(\frac{\alpha}{t}) \cdot \opt
\end{equation*}
If $\gamma$ is $\alpha^{\frac{1}{2t}}$, then $t \leq \frac{\log \alpha}{2}$, and we have:
\begin{equation*}
    \E[w_{S'}]
    = O(\gamma^2 \log_{\gamma}\lambda^{-1}) \cdot \opt
    = O(\alpha^{\frac{1}{t}} t) \cdot \opt
    = O(\frac{\alpha}{t}) \cdot \opt
\end{equation*}

The algorithm of \cref{lem:partial-ds} runs in $O(\frac{\log \Delta}{\eps}) = O(t \log \Delta)$ rounds and the algorithm of \cref{lem:small-sets} runs in $$O(\log_{\gamma} \lambda^{-1} \log_\gamma \Delta) = O(\log_{\alpha^{\frac{1}{2t}}} \alpha^2 \log_\gamma \Delta) = O(t \log \Delta)$$ rounds.
\end{proof}

%\begin{remark}
%With a similar calculations as in the proof of \cref{thm:rand-mds}, the best expected approximation factor we get is $\alpha + 6\lceil\log 4\alpha(\alpha + 1)\rceil + 1 = \alpha + 12\log\alpha + O(1)$. The algorithm needs $O(\alpha \log \Delta)$ many rounds.
%\end{remark}

\begin{proof}[Proof of \cref{lem:small-sets}]
We construct $S'$ in several steps. For a moment, assume that the problem is unweighted (all nodes have weight one). From what we have discussed before, each node can dominate at most $q = O(\alpha)$ nodes in $V \setminus N^{+}_S$. In the first step of constructing $S'$, we try to reduce $q$ to $\frac{q}{2}$. We call a node heavy if it has at least $\frac{q}{2}$ undominated neighbors. To get rid of heavy nodes, first, we sample each of them with probability $\frac{1}{\Delta+1}$. We add all the sampled nodes to $S'$ and update the set of undominated nodes to $V \setminus N^{+}_{S \cup S'}$. The set of heavy nodes is updated accordingly. Then, we sample each heavy node with probability $\frac{2}{\Delta + 1}$. We repeat this for $O(\log \Delta)$ iterations until the sampling probability becomes 1. This ensures that in the end, there is no heavy node. To show that there are not too many sampled nodes, let $n'$ be the number of undominated nodes before the first iteration and observe that we need at least $\frac{n'}{q}$ nodes to dominate them. It can be shown that after all the iterations, the expected number of sampled nodes is $O(\frac{n'}{q})$. So in $O(\log \Delta)$ rounds and with an additive loss of $O(1)$ in the approximation factor, we can reduce $q$ to $\frac{q}{2}$. Repeating this for $O(\log q) = O(\log \alpha)$ times, we get a set $S'$ in $O(\log \alpha \log \Delta)$ rounds with expected size $O(\log \alpha \cdot \opt)$ such that $S \cup S'$ is a dominating set. A detailed discussion on the parameterized version of this algorithm that works also for the weighted case is given in the following.

We start $S'$ as an empty set. Unlike the previous parts, for each node $u$, we set $X_u = \sum_{v \in N^{+}_u \cap (V \setminus N^{+}_{S \cup S'})} x_v$ to be the summation of packing values of only undominated nodes in $N^{+}_u$.
Let $\Gamma_1 = \{u \not \in S \cup S': X_u \geq \gamma^{-1} w_u\}$. We sample nodes of $\Gamma_1$ and update it for $r = \lceil \log_\gamma (\Delta + 1) \rceil + 1$ iterations. Before the first iteration, we initialize the sampling probability $p$ with $\frac{1}{\Delta + 1}$. Then, we run the following in each iteration:
\begin{enumerate}
    \item Sample each node in the current $\Gamma_1$ with probability $p$.
    \item Add the sampled nodes to $S'$.
    \item Update $X_u$ for each node $u$. That is, compute $$X_u = \sum_{v \in N^{+}_u \cap (V \setminus N^{+}_{S \cup S'})} x_v$$.
    \item Remove all the nodes from $\Gamma_1$ with $X_u < \gamma^{-1} w_u$.
    \item Set $p \leftarrow \min(\gamma p, 1)$.
\end{enumerate}
For each element $v$, we define a random variable $c_v$. If $v$ is already dominated by $S$ or if it is not dominated by $S'$ after all the iterations, we set $c_v$ to zero. Otherwise, let $i$ be the first iteration where $v$ is dominated. We set $c_v$ to be the number of sampled nodes in iteration $i$ that dominates $v$.

\begin{lemma}
\label{lem:bounded-covering}
For each node $v$, we have $\E\left[c_v\right] \leq \gamma + 1$.
\end{lemma}
\begin{proof}
If $v$ is already dominated by $S$, then $c_v$ is always zero. So suppose this is not the case and assume that it is in $d$ sets of $\Gamma_1$ before the first iteration. Let $d = d_1 \geq d_2 \geq \dots \geq d_r$ be the sequence that maximize the expected value of $\E[c_v]$ where $d_i$ is the number of sets in $\Gamma_1$ that contains $v$ in the beginning of iteration $i$. Let $p_i = \frac{\gamma^{i-1}}{\Delta + 1}$ be the sampling probablity at iteration $i$. We have:
\begin{equation*}
    \E[c_v] \leq \sum_{i=1}^r p_i d_i \prod_{j=1}^{i-1} (1 - p_j)^{d_j} \leq \sum_{i=1}^r p_i d_i \prod_{j=1}^{i-1} e^{-p_j d_j}
\end{equation*}
To simplify the notation, we define $\beta_i = p_i d_i$. Denote the prefix sum of sequence $\beta_i$ as $\bar{\beta}_i = \sum_{j=1}^{i-1} \beta_i$. Let us emphasize that $\bar{\beta}_i$ does not include $\beta_i$. Rewriting the above inequality using $\beta$ and $\bar{\beta}$, we have:
\begin{equation*}
    \E[c_v] \leq \sum_{i=1}^{r} \beta_i \prod_{j=1}^{i-1} e^{-\beta_j} = \sum_{i=1}^{r} \beta_i e^{-\bar{\beta}_i}
\end{equation*}
Since the sequence $d_1, \dots, d_r$ is non-increasing and $p_i = \gamma p_{i-1}$, we have $\beta_i \leq \gamma \beta_{i-1}$ for any $2 \leq i$. Using this, we can deduce the following:
\begin{align*}
\E[c_v] \leq \sum_{i=1}^{r} \beta_i e^{-\bar{\beta}_i}
&= \beta_1 + \sum_{i=2}^{r} \beta_i e^{-\bar{\beta}_i}\\
&\leq \frac{d_1}{\Delta + 1} + \sum_{i=2}^{r} \gamma \beta_{i-1} e^{-\bar{\beta}_i}\\
&\leq 1 + \gamma \sum_{i=2}^{r} \int_{\bar{\beta}_{i-1}}^{\bar{\beta}_{i}} e^{-x} dx\\
\end{align*}
In the last inequality, note that $\bar{\beta}_i - \bar{\beta}_{i-1} = \beta_{i-1}$ by definition. Since all the integral ranges are disjoint, we have:
\begin{equation*}
\E[c_e] \leq 1 + \gamma \int_{0}^{\infty} e^{-x} dx = \gamma + 1
\end{equation*}
\end{proof}

\begin{lemma}
\label{lem:expected-weight}
The expected total weight of sampled nodes in all iterations is at most $\gamma(\gamma + 1) \cdot \opt$.
\end{lemma}
\begin{proof}
Suppose a node $u$ that is sampled in iteration $i$. Let $T_u$ be the set of nodes in $N^{+}_u$ that are undominated in the beginning of iteration $i$. At that point, $u$ is in $\Gamma_1$ implying $\gamma^{-1} w_u \leq \sum_{v \in T_u} x_v$. So we can upper bound the weight of any sampled node $u$ by $w_u \leq \gamma \sum_{v \in T_u} x_v$. From the definition of $c_v$, each node $v$ appears in $c_v$ many $T_u$s for a sampled node $T_u$. So the total weight of sampled nodes is upper bounded by $\gamma \sum_{v \in V} c_v x_v$. Finally, we know that $\E[c_v] \leq \gamma + 1$ from \cref{lem:bounded-covering} which concludes the proof.
\end{proof}

The set $S \cup S'$ is not necessarily a dominating set. Consider the subproblem of dominating $V \setminus N^{+}_{S \cup S'}$. Clearly, we can dominate these nodes with a set of weight at most $\opt$. Moreover, $\{x_v\}_{v \not \in N^{+}_{S\cup S'}}$ is a feasible packing for this subproblem. Multiply the packing value of each node $v \not \in S \cup S'$ by a factor $\gamma$ and update $X_u$ for each $u \not \in S \cup  S'$. The packing for this subproblem remains feasible. This is because, in the last iteration, we sample all nodes of $\Gamma_1$. So at the end of this iteration, a node $u$ that is not in $S \cup S'$ has $X_u \leq \gamma^{-1}w_u$ and as a result, multiplying the packing values by a factor $\gamma$ is safe.

Now, we define $\Gamma_2$ as $\{u \not \in S \cup S': X_u \geq \gamma^{-1} w_u\}$ and run the procedure with $\Gamma_2$. We repeat this for $t = \lceil\log_{\gamma} \lambda^{-1}\rceil$ times and claim that after that, $S \cup S'$ is a dominating set.

Suppose it is not and let $v$ be an undominated node. At the very beginning, $x_v \geq \lambda \tau_v$ due to property \ref{item:partial-ds-prop-b} of \cref{lem:partial-ds}. Since $v$ is undominated, its packing value is $\gamma^{i} x_v$ after we finish the process of $\Gamma_i$. Since $t = \lceil\log_{\gamma} \lambda^{-1}\rceil$, there should be an $i$ such that the packing value of $v$ is at least $\gamma^{-1} \tau_v$ when we finish working with $\Gamma_{i-1}$. On the other hand and from the definition of $\tau_v$, there is a neighbor $u$ of $v$ with weight $\tau_v$. This means that $u$ is in $\Gamma_i$ and it remains in $\Gamma_i$ until it is sampled. This contradicts that $v$ is not dominated. 

From \cref{lem:expected-weight}, the expected weight of the sampled nodes in one phase is $\gamma(\gamma + 1)\cdot \opt$. So in total, the expected weight of $S'$ is $\gamma(\gamma+1)\lceil \log_{\gamma} \lambda^{-1}\rceil \cdot \opt$.

Each iteration of each phase can be run in $O(1)$ rounds in the \CONGEST model. So in total, the algorithm needs $$O(t \cdot r) = O(\log_\gamma \lambda^{-1} \log_\gamma \Delta)$$ many rounds.
\end{proof}

\ThmRandGeneralMDS*

\begin{proof}
In \cref{lem:small-sets}, we assume $S$ is empty and set $\lambda$ to $\frac{1}{\Delta + 1}$. This does not violate any condition in the lemma. By setting $\gamma$ to $\Delta^{\frac{1}{k}}$, the output $S'$ of the lemma is a dominating set with the claimed size.
\end{proof}

\section{Lower Bound} \label{sec:LB}

\ThmLB*

%\begin{theorem}
%\label{thm:lower-bound}
%Any distributed algorithm that computes any constant approximation of the minimum dominating set on graphs of arboricity 2 requires $\Omega\left(\frac{\log \Delta}{\log\log \Delta}\right)$ rounds in the \LOCAL model.
%\end{theorem}

\begin{proof} 
Kuhn, Moscibroda and Wattenhofer proved that obtaining a constant or poly-logarithmic approximation for the Minimum Vertex Cover (MVC) problem requires $\Omega\left(\frac{\log \Delta}{\log\log \Delta}\right)$ rounds in the \LOCAL model \cite{KMW}. In fact, their lower bound holds even for the \emph{fractional} version of the problem, in which the goal is to assign a value $x_v$ for each node such that for each edge $\{u,v\} \in E$, we have that $x_v + x_u \geq 1$, and $\sum_v x_v$ is minimized.
Let $G$ be the Kuhn-Moscibroda-Wattenhofer (KMW) lower bound graph for approximating the Minimum Fractional Vertex Cover ($\mfvc$), with maximum degree $\Delta$, $n$ nodes, and $m$ edges. We build a new graph $H$ as follows: Take $\Delta^2$ copies of $G$, let us call them $G_1$, $G_2$, $\ldots$, $G_{\Delta^2}$. Add a set $T$ of $n$ additional nodes, one for each node of $G$, and connect each new node to all copies of that original $G$-node. So the degree of each node in $T$ is $\Delta^2$. Next, for each $G_i$, add a node in the middle of each of its edges. This completes the construction of $H$. See \cref{fig:lowerbound} for an illustration. Some observations on $H$ is in the following:
\begin{itemize}[label={-}]
    \item Each $G_i$ has $n + m$ nodes (one copy for each node of $G$ and $m$ middle nodes) and $2m$ edges. Taking $T$ into consideration, $H$ has $\Delta^2(n + m) + n$ nodes and $\Delta^2 (2m + n)$ edges.
    \item The maximum degree of $H$ is $\Delta^2$.
    \item The arboricity of $H$ is 2. For each middle node, orient its two incident edges outward. For each node in $T$, orient its $\Delta^2$ incident edges inward. This gives us an orientation of all the edges of $H$. There is no directed cycle in this orientation and its maximum out-degree is $2$, so the arboricity of $H$ is $2$.
    \item Let $\opt_{\mds}^{H}$ be the size of the minimum dominating set of $H$ and $\opt_{\mvc}^{G}$ ($\opt_{\mfvc}^{G}$) be the size of the minimum (fractional) vertex cover of $G$, then:
    \begin{equation}
    \label{eq:lb}
        \opt_{\mds}^{H} \leq
        \Delta^2 \cdot \opt^{G}_{\mvc} + n = \Delta^2 \cdot \opt^{G}_{\mfvc} + n
    \end{equation}
    The first inequality is because $T$ along with copies of a minimum vertex cover of $G$ in each $G_i$ is a dominating set of $H$. This is because every node in $G_i$ that is not a middle-node has a neighbor in $T$. On the other hand, all the middle-nodes are dominated as we add a vertex cover for each $G_i$. For the equality $\opt^{G}_{\mvc} = \opt^{G}_{\mfvc}$ in \eqref{eq:lb}, we leverage the fact that the KMW graph $G$ is bipartite and as a result the integrality gap of vertex cover on $G$ is 1.\footnote{Alternatively, we can use the fact that $\opt^{G}_{\mvc} \leq 2 \cdot \opt^{G}_{\mfvc}$ in any graph, which would only change the constant in our analysis.}
    In addition, $\opt^{G}_{\mfvc}  \geq \frac{m}{\Delta}$. This holds as if $\{x_v\}_v$ is the optimal fractional solution, we have that $$m \leq \sum_{\{u,v\} \in E} (x_u + x_v) \leq \Delta \sum_v x_v = \Delta \cdot \opt^{G}_{\mfvc}.$$ Since $\opt^{G}_{\mfvc}  \geq \frac{m}{\Delta}$ and for the KMW graph we have $m \geq n$, we can write: $$\opt_{\mds}^{H} \leq (\Delta^2 + \Delta) \cdot \opt^{G}_{\mfvc}.$$
\end{itemize}

\begin{figure}[t]
\centering
\includegraphics[width=0.9\columnwidth]{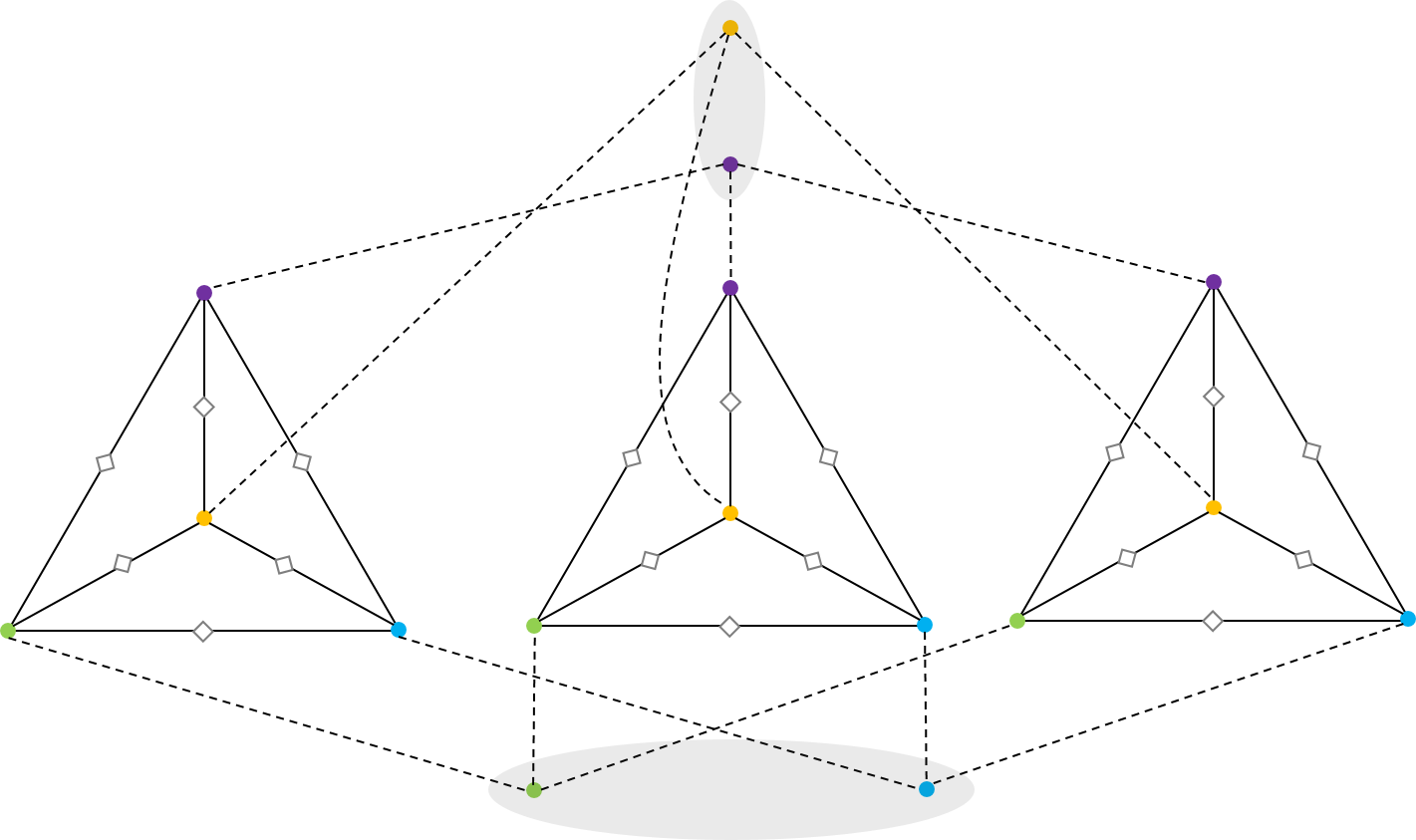}
\caption{The lower bound graph $H$ assumes that the KMW graph $G$ is $K_4$. Only three copies (rather than nine copies) of $G$ are drawn. The middle nodes are indicated by hollow diamonds and $T$ is the set of four nodes in the gray area.}
\label{fig:lowerbound}
\end{figure}

To prove the theorem, suppose there is an algorithm $\mathcal{A}$ with round complexity $o\left(\frac{\log \Delta}{\log\log \Delta}\right)$ that computes a $c$-approximation of Minimum Dominating Set ($\mds$) on $H$ for some constant or poly-logarithmic $c$. Using this algorithm, we will show that we get an algorithm that computes a $c(1 + \frac{1}{\Delta})$ approximation of $\mfvc$ of $G$ in the same number of rounds, hence putting us in contradiction with the lower bound of \cite{KMW}, thus completing the proof.

Notice that the graph $G$ can  simulate the graph $H$ in the \LOCAL model, where each node simulates all its copies in $H$, and each node in $H$ that corresponds to an edge in $G$ is simulated by one of its endpoints. Note that for each edge in $H$, its two endpoints are simulated either by the same node in $G$ or by two neighboring nodes in $G$, hence we can simulate a round of an algorithm in $H$ in one round in $G$. Let us run the algorithm $\mathcal{}$ on the graph $H$, by actually working on the graph $G$. Let $S$ be the computed dominated set of $H$. We want to turn $S$ into a fractional vertex cover of $G$. First, replace each middle node in $S$ with one of its endpoints. This can only decrease the size of $S$ and after it, all the middle nodes are still dominated. Let $S_i$ be the set of nodes of $G_i$ in $S$. Since $S$ dominates all the middle nodes, $S_i$ is a vertex cover of the original graph $G$. As $\mathcal{A}$ is a $c$-approximation algorithm for MDS, we have: $$ \sum_{i \in [\Delta^2]} | S_i | \leq |S| \leq c(\Delta^2 + \Delta) \cdot \opt^{G}_{\mfvc}$$

At the end, each node $v$ of $G$, computes $y_v = \frac{|\{i: v \in S_i\}|}{\Delta^2}$. Observe that $\sum_{v} y_v \leq c(1 + \frac{1}{\Delta}) \cdot \opt^{G}_{\mfvc}$, and we show that $\{y_v\}_v$ is a fractional vertex cover of $G$. This follows from the fact that each $S_i$ is a vertex cover of $G$. Hence, for each edge $\{u,v\} \in G$, if we denote by $\{u_i,v_i\}$ the corresponding edge in $G_i$, we have that at least one of $u_i,v_i$ is in the vertex cover. Let $y^i_v$ indicate if $v_i \in S_i$. We have $$y_u + y_v = \frac{1}{\Delta^2} \sum_{1 \leq i \leq \Delta^2} (y^i_u + y^i_v) \geq 1,$$ as needed. %as each $S_i$ is a vertex cover of $G$ and its size is $\sum_{v} y_v \leq c(1 + \frac{1}{\Delta}) \cdot \opt^{G}_{\mfvc}$. 
So if $\mathcal{A}$ exists, then we can compute a $c(1 + \frac{1}{\Delta})$ approximation of minimum fractional vertex cover of $G$ in $o\left(\frac{\log \Delta}{\log\log \Delta}\right)$ rounds contradicting the lower bound of \cite{KMW}.
\end{proof}

\section*{Acknowledgments}
This work was supported in part by funding from the European Research Council (ERC) under the European Union’s Horizon 2020 research and innovation programme (grant agreement No. 853109), and the Swiss National Foundation (project grant 200021-184735).

\bibliography{refs}

\appendix

\section{Missing Proofs} \label{app_miss_proof}

\begin{observation}
There is a single-round algorithm that computes a $3$-approximation of minimum unweighted dominating set on graphs of arboricity $1$.
\end{observation}

\begin{proof}
Take all non-leaf nodes as the dominating set. This is a dominating set, we next show the approximation ratio. Let $S^*$ be an optimal dominating set. A graph of arboricity 1 is a forest, we can fix a root in each one of the trees. We create a dominating set $S'$ as follows. For each node $v \in S^*$, we add to $S'$ the node $v$ along with its parent and grandparent in the tree (if they exist). Clearly $|S'| \leq 3|S^*|$. We show that $S'$ contains all non-leaf nodes, completing our proof. Assume to the contrary that there is a node $v \not \in S'$ that is an internal node in the tree, and let $u$ be a child of $v$ (that exists, as $v$ is non-leaf). It follows that $u \not \in S^*$, as otherwise $v \in S'$. Since $v$ and $u$ are not in $S^*$, and $S^*$ is a dominating set, there is a child $w$ of $u$ such that $w \in S^*$, but then $v \in S'$, a contradiction.
\end{proof}

%\michal{edit proof to make it self-contained}

%\begin{proof}
%Take all non-leaf nodes as the dominating set. This is the minimum connected dominating set of the tree. Moreover, since any minimum connected dominating set is at most three times larger than the minimum dominating set, the approximation factor is at most $3$ (indeed this approximation is essentially tight for this algorithm as can be seen on a long path).
%\end{proof}

\end{document}